\documentclass{article}
\usepackage[latin9]{inputenc}
\usepackage{geometry}
\usepackage{color}
\usepackage{amsmath}
\usepackage{amsthm}
\usepackage{graphicx}
\usepackage[numbers]{natbib}
\usepackage[ruled]{algorithm2e}
\usepackage{booktabs}
\makeatletter

\numberwithin{equation}{section}
\numberwithin{figure}{section}
\theoremstyle{plain}

\theoremstyle{definition}

\theoremstyle{plain}

\usepackage{url}
\usepackage{xargs}
\usepackage{amsfonts}
\usepackage{amssymb}
\usepackage[colorinlistoftodos,prependcaption,textsize=tiny]{todonotes}
\usepackage[small,bf]{caption}
\usepackage{tikz}
\usepackage{enumitem}
\setlist{nolistsep}
\usepackage{hyperref}
\hypersetup{
    colorlinks,
    citecolor=red,
    filecolor=red,
    linkcolor=red,
    urlcolor=red
}

\usepackage{fancybox}
\usepackage{multirow}
\usepackage{nicefrac}
\usepackage{hyperref}

\expandafter\def\expandafter\normalsize\expandafter{%
    \normalsize
    \setlength\belowdisplayskip{5pt}
    \setlength\belowdisplayshortskip{5pt}
}

\usepackage{tocloft}
\usepackage{titlesec}
\usepackage{fancyhdr}
\usepackage{fix-cm}
\makeatletter
\newcommand\HUGE{\@setfontsize\Huge{38}{47}}
\makeatother

\titleclass{\part}{top}
\titleformat{\part}[display]
  {\normalfont\HUGE\sffamily}{\partname\ \thepart}{0pt}  
  {\titlerule\vskip2pt\titlerule\vskip20pt\HUGE\bfseries\filleft}
\titlespacing*{\part} {0pt}{30pt}{50pt}

\titleformat{\chapter}[display]
  {\normalfont\LARGE\sffamily}{\chaptertitlename\ \thechapter}{0pt}  
  {\titlerule\vskip2pt\titlerule\vskip20pt\LARGE\bfseries\filleft}
\titleformat{\section}
  {\normalfont\Large\bfseries\sffamily}{\rule[.12ex]{8pt}{8pt}~\thesection}{0.5em}{}
\titleformat{\subsection}
  {\normalfont\large\bfseries\sffamily}{\rule[.12ex]{8pt}{8pt}~\thesubsection}{0.5em}{}
\titlespacing*{\chapter} {0pt}{20pt}{30pt}

\usepackage{thmtools}
\usepackage{thm-restate}

\newtheorem{theorem}{Theorem}[section]
\newtheorem{lemma}[theorem]{Lemma}
\theoremstyle{definition}
\newtheorem{definition}{Definition}
\makeatother

\providecommand{\definitionname}{Definition}
\providecommand{\lemmaname}{Lemma}
\providecommand{\theoremname}{Theorem}

\begin{document}
\global\long\def\defeq{\stackrel{\mathrm{{\scriptscriptstyle def}}}{=}}%
\global\long\def\norm#1{\left\Vert #1\right\Vert }%
\global\long\def\R{\mathbb{R}}%
\global\long\def\abs#1{\left\vert #1\right\vert }%
\global\long\def\bdiag{\mathbf{Diag}}%
\global\long\def\nnz{\mathrm{nnz}}%
\global\long\def\mh{\mathbf{H}}%
\global\long\def\mi{\mathbf{I}}%
\global\long\def\mv{\mathbf{V}}%
\global\long\def\mw{\mathbf{W}}%
\global\long\def\ms{\mathbf{S}}%
\global\long\def\mproj{\mathbf{P}}%
\global\long\def\msigma{\mathbf{\Sigma}}%
\global\long\def\ma{\mathbf{A}}%
 
\global\long\def\Rn{\mathbb{R}^{n}}%
\global\long\def\tr{\mathrm{Tr}}%
\global\long\def\poly{\mbox{poly}}%
\global\long\def\diag{\mathrm{diag}}%
\global\long\def\cov{\mathrm{Cov}}%
\global\long\def\E{\mathbb{E}}%
\global\long\def\P{\mathbb{P}}%
\global\long\def\Var{\mathrm{Var}}%
\global\long\def\rank{\mathrm{rank}}%
\global\long\def\Ent{\mathrm{Ent}}%
\global\long\def\vol{\mathrm{vol}}%
\global\long\def\spe{\mathrm{op}}%
\global\long\def\op{\mathrm{op}}%
\global\long\def\H{\mathbf{H}}%

\title{Strong Self-Concordance and Sampling}
\author{Aditi Laddha\thanks{Georgia Tech, aladdha6@gatech.edu}, Yin Tat Lee\thanks{University of Washington and Microsoft Research, yintat@uw.edu},
Santosh S. Vempala\thanks{Georgia Tech, vempala@gatech.edu}}
\maketitle
\begin{abstract}
Motivated by the Dikin walk, we develop aspects of an interior-point
theory for sampling in high dimension. Specifically, we introduce a symmetric parameter
and the notion of strong self-concordance. These properties imply that the corresponding
Dikin walk mixes in $\tilde{O}(n\bar{\nu})$ steps from a warm start
in a convex body in $\mathbb{R}^{n}$ using a strongly self-concordant barrier
with symmetric self-concordance parameter $\bar{\nu}$. For many natural
barriers, $\bar{\nu}$ is roughly bounded by $\nu$, the standard
self-concordance parameter. We show that this property and strong
self-concordance hold for the Lee-Sidford barrier. As a consequence,
we obtain the first walk to mix in $\tilde{O}(n^{2})$ steps for an
arbitrary polytope in $\mathbb{R}^{n}$. Strong self-concordance for other
barriers leads to an interesting (and unexpected) connection ---
for the universal and entropic barriers, it is implied by the KLS
conjecture.
\end{abstract}

\section{Introduction}

The interior-point method is one of the major successes of optimization,
in theory and practice \citep{karmarkar1984new,renegar1988polynomial,vaidya1989speeding}.
It has led to the currently asymptotically fastest algorithms for solving
linear and semidefinite programs and is a popular method for the
accurate solution of medium to large-sized instances. The results of Nesterov
and Nemirovski \citep{nesterov1994interior} demonstrate that $\nu=O(n)$
is possible for any convex set using their universal barrier, where $\nu$ is the self-concordance parameter of the barrier. For
linear programming with feasible region $\left\{ x\,:\,Ax\ge b\right\} $,
the simple logarithmic barrier $g(x)=-\sum_{i}\ln((Ax-b)_{i})$ has
$\nu=O(m)$ for an $m\times n$ constraint matrix $A$, and is efficiently
computable (the universal barrier is polytime to estimate, but requires
the computation of volume of a convex body). In progress over the
past decade, Lee and Sidford \citep{lee2014path,lee2015efficient,lee2019solving}
introduced a barrier for linear programming that achieves $\nu=O(n\log^{O(1)}(m))$
while being efficiently computable. The interior-point method has
also directly influenced the design of combinatorial algorithms, leading
to faster methods for maxflow/mincut and other optimization problems
\citep{madry2010fast,christiano2011electrical,sherman2013nearly,madry2013navigating,kelner2014almost,peng2016approximate,sherman2017area}.

Sampling convex bodies is a fundamental problem that has close connections
to convex optimization. Indeed, convex optimization can be reduced
to sampling \citep{bertsimas2004solving}. The most general methods that lead
to polynomial-time sampling algorithms are the ball walk and hit-and-run,
both requiring only membership oracle access to the convex set being sampled. These
methods are not affine-invariant, i.e., their complexity depends on
the affine position of the convex set. A tight bound on their complexity
is $O^{*}\left(n^{2}R^{2}/r^{2}\right)$ where the convex body contains
a ball of radius $r$ and is mostly contained in a ball of radius
$R$ \citep{kannan1997random,lovasz2007geometry,lovasz2006hit,lovasz2006fast}. The ratio $R/r$ can be made $O(\sqrt{n})$
for any convex body by a suitable affine transformation.
This effectively makes the complexity $O^{*}(n^{3})$. However, the rounding (e.g.,
by near-isotropic transformation) is an expensive step, and its current
best complexity is $O^{*}(n^{4})$ \citep{lovasz2006simulated}. Even for polytopes,
this the rounding/isotropic step takes $O(mn^{4.5})$ total time for a
polytope with $m$ inequalities using an improved amortized analysis
of the per-step complexity \citep{mangoubi2019faster}.

Interior-point theory offers an alternative sampling method with no
need for rounding. A convex barrier function, via its Hessian, naturally
defines an ellipsoid centered at each interior point of a convex body,
the \emph{Dikin }ellipsoid, which is always contained in the body.
The Dikin walk, at each step, picks a uniformly random point in the
Dikin ellipsoid around the current point. To ensure a uniform stationary
density, the new point is accepted with a probability that depends on
the ratio of the volumes of the Dikin ellipsoids at the two points,
see Algorithm \ref{algo:Dikin} below. Kannan and Narayanan \citep{kannan2012random}
showed that the mixing rate of this walk with the standard logarithmic
barrier is $O(mn)$ for a polytope in $\R^{n}$ defined using $m$
inequalities. Each step of the walk involves computing the determinant and can be done in time $O(mn^{\omega-1})$, leading to an overall
arithmetic complexity of $O(m^{2}n^{\omega})$ (see also \citep{sachdeva2016mixing}
for a shorter proof of a Gaussian variant). Using a different more continuous
approach, where each step is the solution of an ODE (rather than a
straight-line step), Lee and Vempala \citep{lee2018convergence} showed
that the Riemannian Hamiltonian Monte Carlo improves the mixing rate
for polytopes to $O(mn^{2/3})$ while keeping the same per-step complexity.
This leads to the following basic questions:
\begin{itemize}
\item What is the fastest possible mixing rate of a Dikin walk?
\item Is a mixing rate of $O(n)$ possible while keeping each step efficient (say
matrix multiplication time or less)?
\end{itemize}
These are the natural analogies to the progress in optimization, where
for the first, Nesterov and Nemirovski show a convergence rate to
the optimum of $O(\sqrt{n})$, and for the second, Lee and Sidford
show $\tilde{O}(\sqrt{n})$ for linear programming while maintaining
efficiency.

These questions, in the context of sampling, lead to new challenges.
Whereas for optimization, one step can be viewed as moving to the
optimum of the objective in the current Dikin ellipsoid (a Newton
step), for sampling, the next step is a random point in the Dikin
ellipsoid; and since these ellipsoids have widely varying volumes,
maintaining the correct stationary distribution takes some work. In
particular, one needs to show that with large probability, the Dikin
ellipsoids at the current point and proposed next point have volumes
within a constant factor; this would imply that a standard Metropolis
filter succeeds with large probability and there is no ``local''
conductance bottleneck. For global convergence, the two important
ingredients are showing that one-step distributions from nearby points
have a large overlap and a suitable isoperimetric inequality. Both
parts depart significantly from the Euclidean set-up as the notion
of distance is defined by local Dikin ellipsoids.

To address these challenges, in place of the
self-concordance parameter $\nu$, we have a symmetric self-concordance
parameter $\bar{\nu}$. It is the smallest number such that for any
point $u$ in a convex body $K$, with Dikin ellipsoid $E_{u}$, we have $E_{u}\subseteq K\cap(2u-K)\subseteq\sqrt{\bar{\nu}}E_{u}$.
In general $\bar{\nu}$ can be as high as $\nu^{2}$ but for some
important barriers, it is bounded as $O(\nu)$. This includes the
logarithmic barrier, and, as we show, the Lee-Sidford barrier. This
definition and parameter allows us to show that the isoperimetric
(Cheeger) constant for the Dikin distance is asymptotically at least
$1/\sqrt{\bar{\nu}}$.

We need a further, important refinement. The notion of self-concordance
itself bounds the rate of change of the Hessian of the barrier (i.e.,
the Dikin matrix) with respect to the local metric in the spectral
norm, i.e., the maximum change in any direction. We define \emph{strong
}self-concordance as the requirement that this derivative be bounded
in \emph{Frobenius }norm. Again, the logarithmic barrier satisfies
this property, and we show that the Lee-Sidford barrier does as well.

Our main general result then is that the Dikin walk defined using
any symmetric, strongly self-concordant barrier with convex Hessian
mixes in $O(n\bar{\nu})$ steps. We prove that the LS barrier satisfies
all these conditions with $\bar{\nu}=\tilde{O}(n)$ and so has a mixing
rate of $\tilde{O}(n^{2})$ for polytopes, completely answering the
second question, and improving on several existing bounds in \citep{chen2018fast,gustafson2018john}.
We also show that the Dikin walk with the standard logarithmic barrier
can be implemented in time $O(nnz(A)+n^{2})$ where $nnz(A)$ is the
number of nonzero entries in the constraint matrix $A$. This answers
the open question posed in \citep{lee2015efficient,kannan2012random}.
These results along with earlier work on sampling polytopes are collected
in Table \ref{tab:sampling}. We note that while for the Dikin walk
with a logarithmic barrier, there are simple examples showing that
the mixing rate of $O(mn)$ is tight (take a hypercube and duplicate
one of its facets $m-n$ times), for the Dikin walk with the LS barrier,
we are not aware of a tight example or one with mixing rate greater
than $\tilde{O}(n)$. There is the tantalizing possibility
that it mixes in nearly linear time. Thus, the overall arithmetic complexity for sampling a polytope is reduced to $m\cdot\min\left\{ \textrm{nnz}(A)\cdot n+n^{3}, n^{\omega+1}\right\}$ which improves the state of the art for all ranges of $m$.
\begin{table}
\centering
\caption{\label{tab:sampling}The complexity of uniform polytope sampling from
a warm start.}
\label{tab:compare}
\begin{tabular}{ccl}
\hline 
\toprule
Markov Chain & Mixing Rate & Per step cost\\
\midrule
Ball Walk{\footnotemark}\citep{kannan1997random} & $n^{2}R^{2}/r^{2}$ & $mn$\\
\hline 
Hit-and-Run{\footnotemark[1]}\citep{lovasz2006hit} & $n^{2}R^{2}/r^{2}$ & $mn$\\
\hline 
Dikin \citep{kannan2012random} & $mn$ & $mn^{\omega-1}$\\
\hline 
RHMC \citep{lee2018convergence} & $mn^{\frac{2}{3}}$ & $mn^{\omega-1}$\\
\hline 
Geodesic Walk\citep{lee2017geodesic} & $mn^{\frac{3}{4}}$ & $mn^{\omega-1}$\\
\hline 
John's Walk\citep{gustafson2018john} & $n^{7}$ & $mn^{4}+n^{8}$\\
\hline 
Vaidya Walk\citep{chen2018fast} & $m^{\frac{1}{2}}n^{\frac{3}{2}}$ & $mn^{\omega-1}$\\
\hline 
Approximate John Walk\citep{chen2018fast} & $n^{2.5}$ & $mn^{\omega-1}$\\
\hline 
\textbf{\textcolor{red}{Dikin (this paper)}} & $mn$ & \textbf{\textcolor{red}{$\nnz(A)+n^{2}$}}\\
\hline 
\textbf{\textcolor{red}{Weighted Dikin (this paper)}} & \textbf{\textcolor{red}{$n^{2}$}} & \textcolor{red}{$mn^{\omega-1}$}\\
\bottomrule
\end{tabular}
\end{table}

\footnotetext[1]{These entries are for general convex bodies
presented by oracles, with $R/r$ measuring the \emph{roundness} of
the input body; this can be made $O(\sqrt{n})$ with a rounding procedure
that takes $n^{4}$ steps (membership queries). \emph{After} rounding,
the amortized per-step complexity of the ball walk in a polytope is
$\tilde{O}(m)$.}
We also study the notions of symmetry and strong self-concorda\-nce introduced in this paper for three well-studied barriers, namely,
the classical universal barrier \citep{nesterov1994interior}, the
entropic barrier \citep{bubeck2014entropic} and the canonical barrier
\citep{hildebrand2014canonical}. While these barriers are not particularly
efficient to evaluate, they are interesting because they all achieve
the best (or nearly best) possible self-concordance parameter values
for arbitrary convex sets and convex cones (for the canonical barrier),
and have played an important role in shaping the theory of interior-point
methods for optimization. For the canonical barrier, the work of Hildebrand
already establishes the convexity of the log determinant function (by
definition of the barrier), and strong self-concordance \citep{hildebrand2014canonical}.
For the entropic and universal barriers, we present an unexpected
connection: the strong self-concordance is implied by the KLS isoperimetry
conjecture! This suggests the possibility of more fruitful connections
yet to be discovered using the notion of strong self-concordance.
\subsection{Dikin Walk}

The general Dikin walk is defined as follows. For a convex set $K$
with a positive definite matrix $\mh(u)$ for each point $u\in K$,
let
\[
E_{u}(r)=\left\{ x:\,(x-u)^{\top}\mh(u)(x-u)\le r^{2}\right\} .
\]

\begin{algorithm}
\SetAlgoLined
\DontPrintSemicolon
\SetKwData{Left}{left}\SetKwData{This}{this}\SetKwData{Up}{up}
\SetKwFunction{Union}{Union}\SetKwFunction{FindCompress}{FindCompress}
\SetKwInOut{Input}{input}\SetKwInOut{Output}{output}
\Input{starting point $x_{0}$ in a polytope $P=\left\{ x:\,\mathbf{A}x\ge b\right\} $}
\Output{$x_T$}
Set $r=\frac{1}{512}$\;
\For{$t\leftarrow 1$ \KwTo $T$}{ 
$x_t \leftarrow x_{t-1}$\;
Pick $y$ from $E_{x_t}(r)$\;
 $x_t \leftarrow y$ with probability $\min\left\{ 1,\frac{\vol(E_{x_t}(r))}{\vol(E_{y}(r))}\right\} $\;
}
\caption{$\mathtt{DikinWalk}$}
\label{algo:Dikin}
\end{algorithm}

\subsection{Strong Self-Concordance}

We require a family of matrices to have the following properties.
Usually but not necessarily, these matrices come from the Hessian
of some convex function.
\begin{definition}[Self-concordance]
For any convex set $K\subseteq\Rn$, we call a matrix function $\mathbf{H}:K\rightarrow\R^{n\times n}$ self-concordant if for any $x\in K$, we have
\[
-2\Vert h\Vert_{\mathbf{H}(x)}\mathbf{H}(x)\preceq\frac{d}{dt}\mathbf{H}(x+th)\preceq2\Vert h\Vert_{\mathbf{H}(x)}\mathbf{H}(x).
\]
\end{definition}
\begin{definition}[$\bar{\nu}$-Symmetry]
For any convex set $K\subseteq\Rn$, we call a matrix function $\mathbf{H}:K\rightarrow\R^{n\times n}$ $\bar{\nu}$-symmetric if for any $x\in K$, we have
\[
E_{x}(1)\subseteq K\cap(2x-K)\subseteq E_{x}(\sqrt{\bar{\nu}}).
\]
\begin{figure}[h]
 \centering
\includegraphics[width=0.5\linewidth]{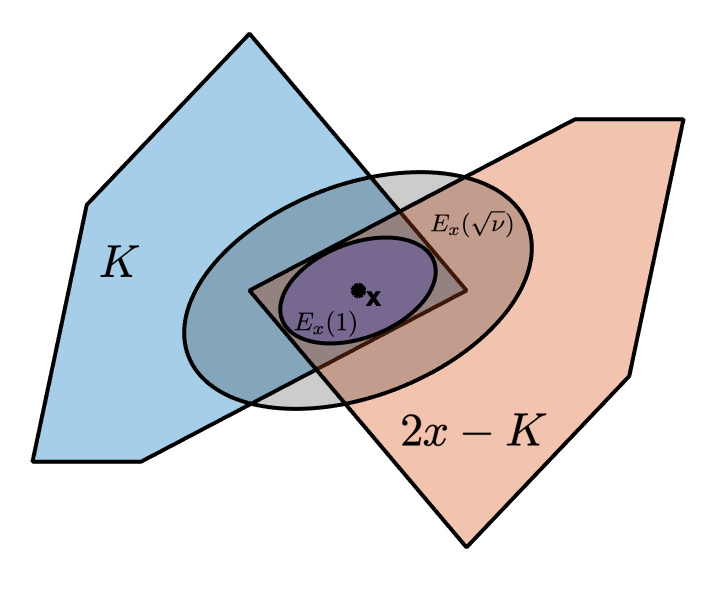}
\caption{$E_{u}(1)\subseteq K\cap(2u-K)\subseteq E_{u}(\sqrt{\bar{\nu}}).$}
\end{figure}
\end{definition}
The following lemma shows that self-concordant matrix functions also enjoy a similar regularity as the usual self-concordant functions.
\begin{lemma}
\label{lem:global_self_concordant}Given any self-concordant matrix
function $\mh$ on $K\subseteq\Rn$, we define $\|v\|_{x}^{2}=v^{\top}\mathbf{H}(x)v$.
Then, for any $x,y\in K$ with $\Vert x-y\Vert_{x}<1$, we have 
\[
\left(1-\Vert x-y\Vert_{x}\right)^{2}\mathbf{H}(x)\preceq\mathbf{H}(y)\preceq\frac{1}{\left(1-\Vert x-y\Vert_{x}\right)^{2}}\mathbf{H}(x).
\]
\end{lemma}
Proof in \ref{proof:2}. Many natural barriers, including the logarithmic barrier and the LS-barrier,
satisfy a much stronger condition than self-concordance,  which we define here.
\begin{definition}[Strong Self-Concordance]
For any convex set $K\subseteq\Rn$, we say a matrix function $\mathbf{H}:K\rightarrow\R^{n\times n}$
is strongly self-concordant if for any $x\in K$, we have
\[
\norm{\mathbf{H}(x)^{-1/2}D\mathbf{H}(x)[h]\mathbf{H}(x)^{-1/2}}_{F}\le2\norm h_{x}
\]
where $D\mathbf{H}(x)[h]$ is the directional derivative of $\mh$
at $x$ in the direction $h$.
\begin{figure}[h]
 \centering
\includegraphics[width=0.5\linewidth]{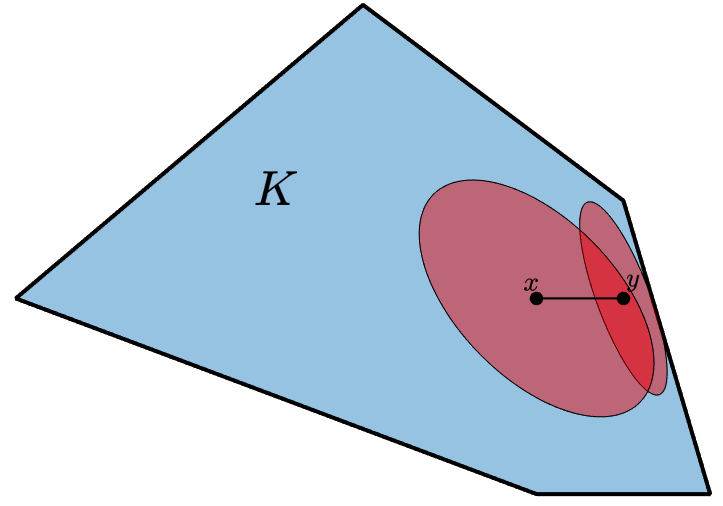}
\caption{Strong self-concordance measures the rate of change of Hessian of a barrier in the Frobenius norm}
\end{figure}
\end{definition}
Similar to Lemma \ref{lem:global_self_concordant}, we have a global
version of strong self-concordance.
\begin{lemma}
\label{lem:global_strongly_self_concordant}Given any strongly self-concordant
matrix function $\mh$ on $K\subset\Rn$. For any $x,y\in K$ with
$\Vert x-y\Vert_{x}<1$, we have
\[
\|\mh(x)^{-\frac{1}{2}}(\mh(y)-\mh(x))\mh(x)^{-\frac{1}{2}}\|_{F}\leq\frac{\|x-y\|_{x}}{(1-\|x-y\|_{x})^{2}}.
\]

\end{lemma}

Proof in \ref{proof:4}. We note that strong self-concordance is stronger than self-concordance
since the Frobenius norm is always larger or equal to the spectral
norm. As an example, we will verify that the conditions hold for the
standard log barrier (Lemma \ref{lem:dssc}).

\subsection{Results}
Our first theorem is the following.

\begin{restatable}{theorem}{dikingeneral}\label{lem:dikingeneral}The
mixing rate of the Dikin walk for a $\bar{\nu}$-symmetric, strongly self-concordant
matrix function with convex log determinant is $O(n\bar{\nu})$.\end{restatable}

This implies faster mixing and sampling for polytopes using the LS
barrier (see Sec. \ref{subsec:LS-Barrier} for the definition).
\begin{theorem}
\label{thm:mixing-ls}The mixing rate of the Dikin walk based on the
LS barrier for any polytope in $\R^{n}$ is $\tilde{O}(n^{2})$ and
each step can be implemented in $\tilde{O}(mn^{\omega-1})$\footnote{We use $\tilde{O}$ to hide factors polylogarithmic in $n,m$.}
arithmetic operations. 
\end{theorem}

On a related note, we show that each step of the standard Dikin walk
is fast, and does not need matrix multiplication.
\begin{theorem}
\label{thm:dikin}The Dikin walk with the logarithmic barrier for
a polytope $\{\mathbf{A}x\ge b\}$ can be implemented in time $O(\nnz(\mathbf{A})+n^{2})$
per step while maintaining the mixing rate of $O(mn)$. See \ref{section4}.
\end{theorem}

The next lemma results from studying strong self-concordance for classical
barriers. The KLS constant below is conjectured to be $O(1)$ and
known to be $O(n^{\frac{1}{4}})$ \citep{lee2017eldan}.
\begin{lemma}
\label{lem:SSC-KLS}Let $\psi_{n}$ be the KLS constant of isotropic
logconcave densities in $\R^{n}$, namely, for any isotropic logconcave
density $p$ and any set $S\subset\Rn$, we have
\[
\int_{\partial S}p(x)dx\geq\frac{1}{\psi_{n}}\min\left\{ \int_{S}p(x)dx,\int_{\R^{n}\backslash S}p(x)dx\right\} .
\]
Let $\mh(x)$ be the Hessian of the universal or entropic barriers.
Then, we have
\[
\norm{\mathbf{H}(x)^{-1/2}D\mathbf{H}(x)[h]\mathbf{H}(x)^{-1/2}}_{F}=O(\psi_{n})\norm h_{x}.
\]

In short, the universal and entropic barriers in $\R^{n}$ are strongly
self-concordant up to a scaling factor depending on $\psi_{n}$.
\end{lemma}

In fact, our proof( see Section \ref{section5}) shows that up to a logarithmic factor the strong
self-concordance of these barriers is \emph{equivalent }to the KLS
conjecture.

\section{Mixing with Strong Self-Concordance}
A key ingredient of the proof of Theorem \ref{lem:dikingeneral} is
the following lemma.
\begin{lemma}
\label{lem:dikin-one-step}For two points $x,y\in P$, with $\Vert x-y\Vert_{x}\le\frac{1}{512\sqrt{n}}$,
we have $d_{TV}(P_{x},P_{y})\leq\frac{3}{4}$. 
\end{lemma}
\begin{proof}  Let $\mathcal{E}(x,\mathbf{A})$ denote the uniform distribution
over an ellipsoid centered at $x$ with covariance matrix $\mathbf{A}$ and radius $r=\frac{1}{512}$. Then,
\begin{align}
d_{\mathrm{TV}}(P_{x},P_{y}) & \leq\frac{1}{2}\text{rej}_{x}+\frac{1}{2}\text{rej}_{y}+d_{\textrm{TV}}(\mathcal{E}(x,\mathbf{H}(x)),\mathcal{E}(y,\mathbf{H}(y))\label{eq:def}
\end{align}
where $\text{rej}_{x}$ and $\text{rej}_{y}$ are the rejection probabilities
at $x$ and $y$.
We break the proof into 2 parts. First we bound the rejection probability at $x$. Consider the algorithm picks a point
$z$ from $E_{x}(r)$. Let $f(z)=\ln\det\mathbf{H}(z)$. The acceptance
probability of the sample $z$ is
\begin{equation}
\min\left\{ 1,\frac{\vol(E_{x}(r))}{\vol(E_{z}(r))}\right\} =\min\left\{ 1,\sqrt{\frac{\det(\mathbf{H}(z))}{\det(\mathbf{H}(x))}}\right\} .\label{eq:reject_prob_z}
\end{equation}
By our assumption $f$ is a convex function, and hence
\begin{equation}
\ln\frac{\det(\mathbf{H}(z))}{\det(\mathbf{H}(x))}=f(z)-f(x)\geq\langle\nabla f(x),z-x\rangle.\label{eq:log_det_Hzx}
\end{equation}
\begin{equation}
\langle\nabla f(x),z-x\rangle = \langle\H(x)^{-\frac{1}{2}}\nabla f(x),\H(x)^{-\frac{1}{2}}(z-x)\rangle
\end{equation}
where $z' = \H(x)^{-\frac{1}{2}}z$ is sampled from a ball of radius $r$ centered at $x' = \H(x)^{-\frac{1}{2}}x$, and hence we know that 
\[
\Pr(v^{\top}(z'-x')\geq-\epsilon r\|v\|_{2})\geq1-e^{-n\epsilon^{2}/2}.
\]
In particular, with probability at least $0.99$ in $z$, we have
\begin{equation}
\langle\nabla f(x),z-x\rangle\geq-\frac{4r}{\sqrt{n}}\|\H(x)^{-\frac{1}{2}}\nabla f(x)\|_{2}.\label{eq:grad_prob}
\end{equation}
To compute $\|\H(x)^{-\frac{1}{2}} \nabla f(x)\|_{2}^{2}$, it is easier to compute directional
derivative of $\nabla f$. Note that 
\begin{align}
\|\H(x)^{-\frac{1}{2}}\nabla f(x)\|_{2} & = \max_{\|v\|_{2}=1}\left(\H(x)^{-\frac{1}{2}}\nabla f(x)\right)^{\top}v \nonumber\\
 & =\max_{\|v\|_{2}=1}\tr(\mh(x)^{-1}D\mh(x)[\H(x)^{-\frac{1}{2}}v]) \nonumber\\
 & =\max_{\|u\|_{x}=1}\tr\left(\mh(x)^{-\frac{1}{2}}D\mh(x)[u]\mh(x)^{-\frac{1}{2}}\right) \nonumber\\
 & \leq\max_{\|u\|_{x}=1}\sqrt{n}\|\mh(x)^{-\frac{1}{2}}D\mh(x)[u]\mh(x)^{-\frac{1}{2}}\|_{F} \nonumber\\
 & \leq\max_{\|u\|_{x}=1}2\sqrt{n}\|u\|_{x} \leq2\sqrt{n}\label{eq:grad_log_H}
\end{align}
where the first inequality follows from $\left|\sum_{i=1}^{n}\lambda_{i}\right|\leq\sqrt{n}\sqrt{\sum_{i=1}^{n}\lambda_{i}^{2}}$
and the second inequality follows from the definition of strong self-concordance.

Combining \eqref{eq:reject_prob_z}, \eqref{eq:log_det_Hzx}, \eqref{eq:grad_prob}
and \eqref{eq:grad_log_H}, we see that with probability at least
$0.99$ in $z$, the acceptance probability of the sample $z$ is 
\begin{equation}
\min\left\{ 1,\frac{\vol(E_{x}(r))}{\vol(E_{z}(r))}\right\} \geq e^{-4r}\geq0.9922\label{eq:rejection}
\end{equation}
where we used that $r=\frac{1}{512}$. Hence, the rejection
probability $\mathrm{rej}_{x}$ (and similarly $\mathrm{rej}_{y}$)
satisfies
\begin{equation}
\mathrm{rej}_{x}\leq0.0039\qquad\text{and}\qquad\mathrm{rej}_{y}\leq0.0039.\label{eq:rej_xy}
\end{equation}

To bound the second term, note that $d_{\mathrm{TV}}$ follows the triangle inequality. So, we can bound the second term in \eqref{eq:def} as
\begin{align}
d_{\textrm{TV}}(\mathcal{E}(x,\mathbf{H}(x)),\mathcal{E}(y,\mathbf{H}(y)))
 & \leq d_{\textrm{TV}}(\mathcal{E}(x,\mathbf{H}(x)),\mathcal{E}(y,\mathbf{H}(x)))\\
 &\phantom{{}=1} +d_{\textrm{TV}}(\mathcal{E}(y,\mathbf{H}(x)),\mathcal{E}(y,\mathbf{H}(y)))\notag \label{eq:three_terms}
\end{align}
By definition of $d_{\mathrm{TV}}$,
\begin{equation}
    d_{\textrm{TV}}(\mathcal{E}(x,(\mathbf{H}(x)),\mathcal{E}(y,\mathbf{H}(y))) =\frac{1}{2}\frac{\vol(E_{x}\backslash E_{y})}{\vol(E_{x})}+\frac{1}{2}\frac{\vol(E_{y}\backslash E_{x})}{\vol(E_{y})}
\end{equation}
The first term is a ratio of volumes and hence is invariant under the transformation $z \rightarrow \mathbf{H}(x)^{1/2}z$, after which it becomes the total variation distance between 2 balls of radius $r$ whose centers are at a distance at most $\frac{r}{\sqrt{n}}$. To bound this, we use lemma 3.2 from \citep{kannan1997random},
\begin{equation}
d_{\textrm{TV}}(\mathcal{E}(x,\mathbf{H}(x)),\mathcal{E}(y,\mathbf{H}(x))\leq\frac{e}{e+1}\label{eq:vol_distance}
\end{equation}

Now, we bound $d_{\textrm{TV}}(\mathcal{E}(y,\mathbf{H}(x)),\mathcal{E}(y,\mathbf{H}(y)))$.
Let $Y_{x}=\{z:(z-y)^{\top}\mathbf{H}(x)(z-y)\leq r^{2}\}$ and $Y_{y}=\{z:(z-y)^{\top}\mathbf{H}(y)(z-y)\leq r^{2}\}$.
Then,

\begin{align}
d_{\textrm{TV}}(\mathcal{E}(y,(\mathbf{H}(x)),\mathcal{E}(y,\mathbf{H}(y))) & =\frac{1}{2}\frac{\vol(Y_{x}\backslash Y_{y})}{\vol(_{x})}+\frac{1}{2}\frac{\vol(Y_{y}\backslash Y_{x})}{\vol(Y_{y})}\\
      =1&-\frac{1}{2}\frac{\vol(Y_{x}\cap Y_{y})}{\vol(Y_{x})}-\frac{1}{2}\frac{\vol(Y_{x}\cap Y_{y})}{\vol(Y_{y})}\label{eq:dikin_formula} 
 \end{align}
We bound the total variation distance by bounding the fraction of
volume in the intersection of the ellipsoids having the same center. Again, we
can assume that $\mh(y)=\mi$ and that $y = 0$. Then, strong self-concordance
and Lemma \ref{lem:global_strongly_self_concordant} show that
\begin{equation}
\|\mi-\mh(x)^{-1}\|_{F}\leq2\|x-y\|_{x}\leq\frac{1}{256\sqrt{n}}.\label{eq:H_minus_I_F}
\end{equation}
In particular, we have that 
\begin{equation}
\frac{255}{256}\mi\preceq\mh(x)^{-1}\preceq\frac{257}{256}\mi.\label{eq:H_I_op}
\end{equation}

We partition the inverse eigenvalues, $\{\lambda_{i}\}_{i\in [n]}$ of $\mh(x)$ into those
with values at least $1$ and the rest. Then consider the ellipsoid $\mathcal{I}$
whose inverse eigenvalues are $\min\left\{ 1,\lambda_{i}\right\}$ along the eigenvectors of $\mh(x)$. This is contained in both $Y_{x}$ and $Y_{y}$. We will see that
$\vol(\mathcal{I})$ is a constant fraction of the volume of both $Y_{x}$
and $Y_{y}$. First, we compare $\mathcal{I}$ and $Y_{y}$.
\begin{equation}
\begin{split}
    \frac{\vol(Y_x \cap Y_y)}{\vol(Y_{y})} &\geq \frac{\vol(\mathcal{I})}{\vol(Y_{y})}=\left(\prod_{i:\lambda_{i}<1}\lambda_{i}\right)^{1/2}\\
    &=\left(\prod_{i:\lambda_{i}<1}\left(1-(1-\lambda_{i})\right)\right)^{1/2}\\
    &\geq\exp\left(-\sum_{i:\lambda_{i}<1}\left(1-\lambda_{i}\right)\right)\label{eq:vol_E_x}
\end{split}
\end{equation}
where we used that $1-x\geq\exp(-2x)$ for $0\leq x\leq\frac{1}{2}$
and $\lambda_{i}\geq\frac{1}{2}$ \eqref{eq:H_I_op}. From the inequality
\eqref{eq:H_minus_I_F}, it follows that
\[
\sqrt{\sum_{i}(\lambda_{i}-1)^{2}}\le\frac{1}{256\sqrt{n}}.
\]
Therefore, $\sum_{i:\lambda_{i}<1}|\lambda_{i}-1|\leq\frac{1}{256}$.
Putting it into \eqref{eq:vol_E_x}, we have
\begin{equation}
\frac{\vol(Y_{x}\cap Y_{y})}{\vol(Y_{y})}=\frac{\vol(\mathcal{I})}{\vol(Y_{y})}\geq e^{-\frac{1}{256}}.\label{eq:volPxyPx}
\end{equation}
Similarly, we have
\begin{equation}
\begin{split}
\frac{\vol(Y_{x}\cap Y_{y})}{\vol(Y_{x})}&\geq \left(\frac{\prod_{i:\lambda_{i}<1}\lambda_{i}}{\prod_{i:\lambda_{i}}\lambda_{i}}\right)^{1/2}=\left(\frac{1}{\prod_{i:\lambda_{i}>1}\lambda_{i}}\right)^{1/2}\\
 &\geq \left(\frac{1}{\exp(\sum_{i:\lambda_{i}>1}(\lambda_{i}-1))}\right)^{1/2}\geq e^{-\frac{1}{512}}.\label{eq:volPxyPy}
\end{split}
\end{equation}

Putting \eqref{eq:volPxyPx} and \eqref{eq:volPxyPy} into \eqref{eq:dikin_formula},
we have

\begin{equation}
d_{\textrm{TV}}(\mathcal{E}(y,\mathbf{H}(x)),\mathcal{E}(y,\mathbf{H}(y))\leq1-\frac{e^{-\frac{1}{256}}}{2}-\frac{e^{-\frac{1}{512}}}{2}\label{eq:ellipsoid_2}
\end{equation}

Putting \eqref{eq:rej_xy}, \eqref{eq:vol_distance} and \eqref{eq:ellipsoid_2}
into \eqref{eq:def}, we have
\[
d_{\mathrm{TV}}(P_{x},P_{y})\leq\frac{0.0039}{2}+\frac{0.0039}{2}+1-\frac{e^{-\frac{1}{256}}}{2}-\frac{e^{-\frac{1}{512}}}{2}+\dfrac{e}{e+1}\leq\frac{3}{4}
\]
\end{proof}

The next lemma establishes isoperimetry and only needs the symmetric
containment assumption. This isoperimetry is for the cross-ratio distance.
For a convex body $K$, and any two points $x,y\in K$, suppose that
$p,q$ are the endpoints of the chord through $x,y$ in $K$, so that
these points occur in the order $p,x,y,q.$ Then, the \emph{cross-ratio}
distance between $x$ and $y$ is defined as 
\[
d_{K}(x,y)=\frac{\|x-y\|_{2}\|p-q\|_{2}}{\|p-x\|_{2}\|y-q\|_{2}}.
\]
This distance enjoys the following isoperimetric inequality.
\begin{theorem}[\citep{lovasz1999hit}]
\label{thm:dK-iso}For any convex body $K$, and disjoint subsets
$S_{1},S_{2}$ of it, and $S_{3}=K\setminus S_{1}\setminus S_{2},$we
have
\[
\vol(S_{3})\ge d_{K}(S_{1},S_{2})\frac{\vol(S_{1})\vol(S_{2})}{\vol(K)}.
\]
\end{theorem}

We now relate the cross-ratio distance to the ellipsoidal norm.
\begin{lemma}
\label{lem:dikin-dist}For any $x,y\in K$, $d_{K}(x,y)\ge\frac{\Vert x-y\Vert_{x}}{\sqrt{\bar{\nu}}}.$
\end{lemma}
\begin{proof}
Consider the Dikin ellipsoid at $x$. For the chord $[p,q]$ induced by
$x,y$ with these points in the order $p,x,y,q$, suppose that $\|p-x\|_{2}\le\|y-q\|_{2}$.
Then by Lemma \ref{lem:dikingeneral}, $p\in K\cap(2x-K)$. And hence
$\|p-x\|_{x}\le\sqrt{\bar{\nu}}.$ Therefore, 
\begin{align*}
d_{K}(x,y)=\frac{\|x-y\|_{2}\|p-q\|_{2}}{\|p-x\|_{2}\|y-q\|_{2}}&\geq\frac{\|x-y\|_{2}}{\|p-x\|_{2}}\\
&=\frac{\|x-y\|_{x}}{\|p-x\|_{x}}\ge\frac{\|x-y\|_{x}}{\sqrt{\bar{\nu}}}.
\end{align*}
\end{proof}
We can now prove the main conductance bound.
\dikingeneral*
\begin{proof}
We follow the standard high-level outline \citep{vempala2005geometric}.
Consider any measurable subset $S_{1}\subseteq K$ and let $S_{2}=K\setminus S_{1}$
be its complement. Define the points with low escape probability for
these subsets as
\[
S_{i}'=\left\{ x\in S_{i}:\,P_{x}(K\setminus S_{i})<\frac{1}{8}\right\} 
\]
and $S_{3}'=K\setminus S_{1}'\setminus S_{2}'$. Then, for any $u\in S_{1}'$,
$v\in S_{2}'$, we have $d_{TV}(P_{u},P_{v})>1-\frac{1}{4}$. Hence,
by Lemma \ref{lem:dikin-one-step}, we have $\Vert u-v\Vert_{u}\ge\frac{1}{512\sqrt{n}}$.
Therefore, by Lemma \ref{lem:dikin-dist}, 
\[
d_{K}(u,v)\ge\frac{1}{512\sqrt{n}\cdot\sqrt{\bar{\nu}}}.
\]
We can now bound the conductance of $S_{1}$. We may assume that $\vol(S_{i}')\ge\vol(S_{i})/2$;
otherwise, it immediately follows that the conductance of $S_{1}$
is $\Omega(1)$. Assuming this, we have
\begin{align*}
\int_{S_{1}}P_{x}(S_{2})\,dx & \ge\int_{S_{3}'}\frac{1}{8}dx\ge\frac{1}{8}\vol(S_{3}')\\
& \ge\frac{1}{8}d_{K}(S_{1}',S_{2}')\frac{\vol(S_{1}')\vol(S_{2}')}{\vol(P)}\tag*{(from Thm \ref{thm:dK-iso})} \\
 & \ge\frac{1}{32768\sqrt{n\bar{\nu}}}\min\left\{ \vol(S_{1}),\vol(S_{2})\right\} .
\end{align*}
\end{proof}

It is well-known that inverse squared conductance of a Markov Chain is a bound on its
mixing rate, e.g., in the following form.
\begin{theorem}
\citep{lovasz1993random} Let $Q_{t}$ be the distribution of the current
point after $t$ steps of a Markov chain with stationary distribution
$Q$ and conductance at least $\phi$, starting from initial distribution
$Q_{0}$. Then, with $M=\sup_{A}\frac{Q_{0}(A)}{Q(A)}$,

\[
d_{TV}(Q_{t},Q)\leq\sqrt{M}\left(1-\frac{\phi^{2}}{2}\right)^{t}
\]
where $d_{TV}(Q_{t},Q)$ is the total variation distance between $Q_{t}$
and $Q$.
\end{theorem}

\section{Fast Polytope Sampling with the LS barrier}

\subsection{LS Barrier and its Properties\label{subsec:LS-Barrier}}

In this section, we assume the convex set is a polytope $P=\{x\in\mathbb{R}^{n}\vert\mathbf{A}x>b\}$. For any $x\in\mathrm{int}P$, let $\mathbf{S}_{x}=\bdiag(\mathbf{A}x-b)$
and $\mathbf{A}_{x}=\mathbf{S}_{x}^{-1}\mathbf{A}$. We state the definition of the
Lee-Sidford barrier \citep{lee2019solving}, henceforth referred to as LS barrier.
\begin{definition}[LS Barrier]
 The LS barrier is defined as
\[
\psi(x)=\max_{w\in\mathbb{R}^{m}:w\geq0}\frac{1}{2}f(x,w)\]
where \[f(x,w)=\ln\det\left(\mathbf{A}_{x}\mathbf{W}^{1-\frac{2}{q}}\mathbf{A}_{x}\right)-\left(\frac{1}{2}-\frac{1}{q}\right)\sum_{i=1}^{m}w_{i}
\] and $\mathbf{W} = \mathbf{Diag}(w)$, and $q=2(1+\ln m)$.
\end{definition}
We follow the notation in \citep{lee2019solving}:
\begin{definition}
For any $x\in P$, we define $w_{x}=\arg\max_{w\geq0}f(x,w)$, $\mathbf{W}_{x}=\bdiag(w_{x})$,
$s_{x}=\mathbf{A}x-b$, $\mathbf{S}_{x}=\bdiag(s_{x})$, $\mathbf{A}_{x}=\mathbf{S}_{x}^{-1}\mathbf{A}$,
$\mathbf{P}_{x}=\mathbf{W}_{x}^{\frac{1}{2}-\frac{1}{q}}\mathbf{A}_{x}\left(\mathbf{A}_{x}\mathbf{W}_{x}^{1-\frac{2}{q}}\mathbf{A}_{x}\right)^{-1}(\mathbf{W}_{x}^{\frac{1}{2}-\frac{1}{q}}\mathbf{A}_{x})^{\top}$,
$\sigma_{x}=\diag(\mathbf{P}_{x})$, $\mathbf{\Sigma}_{x}=\bdiag(\sigma_{x})$,
$\mathbf{P}_{x}^{(2)}=\mathbf{P}_{x}\circ\mathbf{P}_{x}$, $\mathbf{\Lambda}_{x}=\mathbf{\Sigma}_{x}-\mathbf{P}_{x}^{(2)}$,
$\bar{\mathbf{\Lambda}}_{x}=\mathbf{\Sigma}_{x}^{-1/2}\mathbf{\Lambda}_{x}\mathbf{\Sigma}_{x}^{-1/2}$,
and $\mathbf{N}_{x}=2\bar{\mathbf{\Lambda}}_{x}(\mathbf{I}-(1-\frac{2}{q})\bar{\mathbf{\Lambda}}_{x})^{-1}$.
\end{definition}
\section{Properties of LS Barrier}\label{properties}
\begin{lemma}[\citep{lee2019solving}]
The function $\psi(x)$ has the following properties:
\begin{enumerate}
\item (Lemma 23) $\psi(x)$ is convex.
\item (Lemma 47.2)
\begin{equation}
\mathbf{P}_{x}^{(2)}\preceq\mathbf{\Sigma}_{x}\label{eq:31}
\end{equation}
 
\item (Lemma 31)
\begin{equation}
0\leq\sigma_{x,i}=w_{x,i}\leq1\label{eq:32}
\end{equation}
\begin{equation}
\mathbf{A}_{x}^{\top}\mathbf{W}_{x}\mathbf{A}_{x}\preceq\nabla^{2}\psi(x)\preceq(1+q)\mathbf{A}_{x}^{\top}\mathbf{W}_{x}\mathbf{A}_{x}\label{eq:34}
\end{equation}
\item (Lemma 33) For any $x_{t}=x+th$ and $s_{t}=\mathbf{A}x_{t}-b$, we
have
\begin{equation}
\|\mathbf{S}_{t}^{-1}\frac{d}{dt}s_{t}\|_{\mathbf{W}_{t}}\leq\|h\|_{\nabla^{2}\psi(x_{t})}\label{eq:35}
\end{equation}
\item (Lemma 34) For any $x_{t}=x+th$ and $w_{t}=w_{x_{t}}$, we have
\begin{equation}
\|\mathbf{W}_{t}^{-1}\frac{d}{dt}w_{t}\|_{\mathbf{W}_{t}}\leq q\|h\|_{\nabla^{2}\psi(x_{t})}\label{eq:36}
\end{equation}
\end{enumerate}
\end{lemma}

\subsection{Mixing Rate}
\begin{definition}
The LS matrix for a point $x\in P$ is defined as
\[
\mathbf{H}(x)=(1+q^{2})(1+q)\cdot\mathbf{A}^{\top}\mathbf{S}_{x}^{-1}\mathbf{W}_{x}^{1-\frac{2}{q}}\mathbf{S}_{x}^{-1}\mathbf{A}.
\]
\end{definition}
We establish the strong self-concordance of LS Matrix in the next lemma.
\begin{lemma}[Strong Self Concordance]
\label{lem:LS_strongly}The LS matrix is strongly self-concordant,
i.e., for any $x_{t} \in P$ given by $x_{t}=x+th$ and $\H_t = \H(x_t)$, we
have
\[
\|\mathbf{H}_{t}^{-1/2}(\frac{d}{dt}\mathbf{H}_{t})\mathbf{H}_{t}^{-1/2}\|_{F}\leq2\|h\|_{\mathbf{H}_{t}}.
\]
\end{lemma}
\begin{proof}
We redefine $$\overline{\mh}_{t}=\mathbf{A}^{\top}\mathbf{V}_{t}\mathbf{A}$$
with $\mv_{t}=\mathbf{S}_{t}^{-1}\mathbf{W}_{t}^{1-2/q}\mathbf{S}_{t}^{-1},\; \mathbf{P}_{t}=\sqrt{\mathbf{V}_{t}}\mathbf{A}(\mathbf{A}^{\top}\mathbf{V}_{t}\mathbf{A})^{-1}\mathbf{A}^{\top}\sqrt{\mathbf{V}_{t}}$.
Note that $\mathbf{V}_{t}$ is a diagonal matrix and that $\overline{\mh}_{t}$
and $\mh_{t}$ are just off by a scaling factor. Hence, we have
\begin{align*}
\allowdisplaybreaks
\|\mh_{t}^{-1/2}(\frac{d}{dt}\mh_{t})\mh_{t}^{-1/2}\|_{F}^{2} & =\|\overline{\mh}_{t}^{-1/2}(\frac{d}{dt}\overline{\mh}_{t})\overline{\mh}_{t}^{-1/2}\|_{F}^{2}\\
 & =\tr\overline{\mh}_{t}^{-1}(\frac{d}{dt}\overline{\mh}_{t})\overline{\mh}_{t}^{-1}(\frac{d}{dt}\overline{\mh}_{t})\\
 & =\tr\left((\mathbf{A}^{\top}\mathbf{V}_{t}\mathbf{A})^{-1}\mathbf{A}^{\top}(\frac{d}{dt}\mathbf{V}_{t})\mathbf{A}\right)^2\\
 & =\tr\mathbf{P}_{t}\frac{d\ln\mv_{t}}{dt}\mathbf{P}_{t}\frac{d\ln\mv_{t}}{dt}\\
 & =\frac{d\ln v_{t}}{dt}^{\top}\mathbf{P}_{t}^{(2)}\frac{d\ln v_{t}}{dt}.
\end{align*}
Note that $\mathbf{P}_{t}^{(2)}\preceq\msigma_{t}$, by \eqref{eq:31}.
Therefore,
\begin{align*}
\|\mh_{t}^{-1/2}(\frac{d}{dt}\mh_{t})\mh_{t}^{-1/2}\|_{F}^{2} & \leq\frac{d\ln v_{t}}{dt}^{\top}\msigma_{t}\frac{d\ln v_{t}}{dt}\\
 & =\sum_{i=1}^{m}\sigma_{t,i}\left(\frac{d\ln s_{t,i}^{-2}w_{t,i}^{1-2/q}}{dt}\right)^{2}\\
 & \leq4\sum_{i=1}^{m}\sigma_{t,i}\left(\left(\frac{d\ln s_{t,i}}{dt}\right)^{2}+\left(\frac{d\ln w_{t,i}}{dt}\right)^{2}\right)\\
 & =4\sum_{i=1}^{m}\sigma_{t,i}\left(\left(\frac{1}{s_{t,i}}\frac{ds_{t,i}}{dt}\right)^{2}+\left(\frac{1}{w_{t,i}}\frac{dw_{t,i}}{dt}\right)^{2}\right)\\
 & \leq4(1+q^{2})\|h\|_{\nabla^{2}\psi(x_{t})}^{2}
\end{align*}
where we used $\sigma_{t}=w_{t}$ \eqref{eq:32} in the second
last equation and equations \eqref{eq:35} and \eqref{eq:36} for the last inequality. 

Finally, \eqref{eq:34} shows that $\nabla^{2}\psi(x_{t})\preccurlyeq(1+q)\mathbf{A}_{t}^{\top}\mathbf{W}_{t}\mathbf{A}_{t}$.
Since $0\leq w_{t}=\sigma_{t}\leq1$ by the property of leverage score,
we have
\[
\nabla^{2}\psi(x)\preceq(1+q)\mathbf{A}_{t}^{\top}\mathbf{W}_{t}\mathbf{A}_{t}\preceq(1+q)\mathbf{A}_{t}^{\top}\mathbf{W}_{t}^{1-2/q}\mathbf{A}_{t}=(1+q)\overline{\mh}_{t}.
\]

Thus, $\|h\|_{\nabla^{2}\psi(x_{t})}^{2}\leq(1+q)\|h\|_{\overline{\mh}_{t}}^{2}$.
Hence, we have
\begin{align*}
\|\mathbf{H}_{t}^{-1/2}(\frac{d}{dt}\mathbf{H}_{t})\mathbf{H}_{t}^{-1/2}\|_{F}^{2} & \leq4(1+q^{2})(1+q)\|h\|_{\overline{\mh}_{t}}^{2}\leq4\|h\|_{\mh_{t}}^{2}
\end{align*}
where we used that $\mh_{t}=(1+q^{2})(1+q)\overline{\mh}_{t}$.
\end{proof}
\begin{lemma}\label{lem:symmetry}
The LS-ellipsoid matrix has the following properties:
\begin{enumerate}
\item $\ln\det\mathbf{H}(x)$ is convex.
\item $\mh$ is a symmetric strongly $\bar{\nu}$-self-concordant barrier
with $\bar{\nu}=O(n\log^{3}m)$.
\end{enumerate}
\end{lemma}
\begin{proof}
For any $x\in\mathrm{int}P$, \eqref{eq:32} shows that
\begin{align*}
    \sum_{i}w_{x,i}=\sum_{i}\sigma_{x,i}&=\tr\mathbf{W}_{x}^{\frac{1}{2}-\frac{1}{q}}\mathbf{A}_{x}\left(\mathbf{A}_{x}\mathbf{W}_{x}^{1-\frac{2}{q}}\mathbf{A}_{x}\right)^{-1}(\mathbf{W}_{x}^{\frac{1}{2}-\frac{1}{q}}\mathbf{A}_{x})^{\top}\\
    &=\tr\mi_{n\times n}=n.
\end{align*}
Hence, the LS barrier can be restated as
\begin{align*}
\psi(x) & =\frac{1}{2}\ln\det(\mathbf{A}_{x}^{\top}\mathbf{W}_x^{1-2/q}\mathbf{A}_{x})-\left(\frac{1}{2}-\frac{1}{q}\right)n\\
 & =\frac{1}{2}\ln\det\frac{1}{(1+q^{2})(1+q)}\mathbf{H}(x)-\left(\frac{1}{2}-\frac{1}{q}\right)n
\end{align*}
where $w_{x}$ is the maximizer of $f(x,w)$. Since $\psi(x)$
is convex, so is $\ln\det\mh(x)$. 

Next, we prove that $\bar{\nu}=O(n\log^{3}m)$. For any $x\in P$ and any $y\in E_{x}(1)$, $(y-x)^{\top}\mathbf{A}_{x}\mathbf{W}_{x}^{1-2/q}\mathbf{A}_{x}(y-x)\leq\frac{1}{(1+q^{2})(1+q)}$
and hence
\begin{align*}
&\norm{\mathbf{A}_{x}(y-x)}_{\infty}^{2}\\
&\quad =\max_{i\in[m]}\left(e_{i}^{\top}\mathbf{A}_{x}(\mathbf{A}_{x}\mathbf{W}_{x}^{1-2/q}\mathbf{A}_{x})^{-1/2}(\mathbf{A}_{x}\mathbf{W}_{x}^{1-2/q}\mathbf{A}_{x})^{1/2}(y-x)\right)^{2}\\
 &\quad \leq\frac{1}{(1+q^{2})(1+q)}\max_{i\in[m]}e_{i}^{\top}\mathbf{A}_{x}(\mathbf{A}_{x}\mathbf{W}_{x}^{1-2/q}\mathbf{A}_{x})^{-1}\mathbf{A}_{x}e_{i}\\
 &\quad \leq\max_{i\in[m]}\frac{\sigma_{x,i}}{w_{x,i}^{1-2/q}}\leq\max_{i\in[m]}\frac{\sigma_{x,i}}{w_{x,i}}=1
\end{align*}
since $w_{x,i}\leq1$. So, $E_{x}\subseteq P\cap(2x-P)$
for all $x\in P$. 

For any $y\in P\cap(2x-P)$, we have $\Vert\mathbf{S}_{x}^{-1}\mathbf{A}(x-y)\Vert_{\infty}\leq1$.
Hence, 
\begin{align*}
\frac{(x-y)^{T}\mathbf{H}(x)(x-y)}{(1+q^{2})(1+q)} & =(x-y)^{T}\mathbf{A}^{\top}\mathbf{S}_{x}^{-1}\mathbf{W}_{x}^{1-2/q}\mathbf{S}_{x}^{-1}\mathbf{A}(x-y)\\
 & =\sum_{i=1}^{m}w_{x,i}^{1-2/q}(\mathbf{S}_{x}^{-1}\mathbf{A}(x-y))_{i}^{2} \leq\sum_{i=1}^{m}w_{x,i}^{1-2/q}\\
 &\leq\left(\sum_{i=1}^{m}\left(w_{x,i}^{1-2/q}\right)^{\frac{1}{1-(2/q)}}\right)^{1-2/q}\left(\sum_{i=1}^{m}1^{q/2}\right)^{2/q}\\
 & \le\left(\sum_{i=1}^{m}w_{x,i}\right)^{1-2/q}m^{2/q}\le n^{1-2/q}m^{2/q}\leq en.
\end{align*}
\end{proof}
Lemmas \ref{lem:LS_strongly} and \ref{lem:symmetry} imply that mixing time of Dikin walk with LS matrix is $\tilde{O}(n^2)$ from a warm start.
Implementing each step of this walk involves the following tasks:
\begin{enumerate}
\item Compute $\mathbf{H}(x)^{-1/2}v$ for some vector $v$
\item Compute the ratio $\det(\mathbf{H}(y)^{-1}\mathbf{H}(x))$ for points
$x,y$.
\end{enumerate}
Given $w_{x}$, computing $\mathbf{H}(x)$, its inverse and its determinant can all
be done in time $\tilde{O}\left(mn^{\omega-1}\right)$. $w_{x}$ can be updated in $\tilde{O}(mn^{\omega-1})$ per step as shown in \citep[Theorem 46]{lee2019solving}. Using this, each step of Dikin walk with LS Matrix can be implemented in time $O(mn^{\omega-1})$
This means that the total time to sample a polytope from a warm start is $\tilde{O}(mn^{\omega+1})$ as claimed in Theorem \ref{thm:mixing-ls}.

\section{Fast Implementation of Dikin walk}\label{section4}
\begin{lemma}[Strong Self-Concordance]
\label{lem:dssc} The matrix function $\mathbf{H}(x)=\mathbf{A}^{\top}\mathbf{S}_{x}^{-2}\mathbf{A}$ which is the Hessian of the log barrier function $\phi(x)=-\sum_{i=1}^{m}\log\left(A_{i}x-b_{i}\right)$,
is strongly self-concordant.
\end{lemma}
\begin{proof}
Let $x_{t}=x+th$ for some fixed vector $h$. Let $\mathbf{S}_{t}=\bdiag(\mathbf{A}x_{t}\allowbreak - b)$, $\ma_{t}=\ms_{t}^{-1}\ma$, $\mathbf{P}_{t}=\ma_{t}(\ma_{t}^{\top}\ma_{t})^{-1}\ma_{t}^{\top}$,
$\sigma_{t}=\diag(\mathbf{P}_{t})$, $\mathbf{\Sigma}_{t}=\bdiag(\sigma_{t})$,
and $\mathbf{P}_{t}^{(2)}=\mathbf{P}_{t}\circ\mathbf{P}_{t}$. By
\citep[Lemma 47.2]{lee2019solving}, $\mathbf{P}_{t}^{(2)}\preccurlyeq\mathbf{\Sigma}_{t}\preceq\mi$.
We are now ready to prove strong self-concordance.
\begin{align*}
&\|\mathbf{H}_{t}^{-1/2}(\frac{d}{dt}\mathbf{H}_{t})\mathbf{H}_{t}^{-1/2}\|_{F}^{2} \\
&=\tr\mathbf{H}_{t}^{-1}(\frac{d}{dt}\mathbf{H}_{t})\mathbf{H}_{t}^{-1}(\frac{d}{dt}\mathbf{H}_{t})=\tr\mathbf{P}_{t}\frac{d\ln s_{t}^{-2}}{dt}\mathbf{P}_{t}\frac{d\ln s_{t}^{-2}}{dt}\\
 & =\frac{d\ln s_{t}^{-2}}{dt}^{\top}\mathbf{P}_{t}^{(2)}\frac{d\ln s_{t}^{-2}}{dt}\leq\sum_{i=1}^{m}\left(\frac{d\ln s_{t}^{-2}}{dt}\right)^{2}\\
 & =\sum_{i=1}^{m}4s_{t,i}^{-2}\left(a_{i}^{\top}h\right)^{2}=4h^{\top}\mathbf{A}^{\top}\mathbf{S}_{t}^{-2}\mathbf{A}h=4\norm h_{\mathbf{H}_{t}}^{2}.
\end{align*}
\end{proof}
The function $\log\det\mathbf{A}^{\top}\mathbf{S}_{x}^{-2}\mathbf{A}$ is called the volumetric barrier and is known to be convex.
\begin{lemma}[{\citep[Lemma 3]{vaidya1996new}}]
$f(x)=\log\det\mathbf{A}^{\top}\mathbf{S}_{x}^{-2}\mathbf{A}$ is a convex function in $x$.
\end{lemma}

The main result of this section is to give an even faster implementation
by noting that in fact we can avoid explicitly computing $\mathbf{H}(x)$
or its inverse or determinant for the Dikin walk with log barrier.
This resolves an open problem posed in \citep{kannan2012random,lee2015efficient}.

The main challenge is to avoid computing the determinant of $\mh(x)$.
In fact, what one needs is an unbiased estimator of the ratio of two
such determinants. We reduce this, first to estimating a log-det,
and then to an inverse maintenance problem in the next two lemmas. 

To calculate rejection probability for Dikin Walk, we want an unbiased
estimator of $\frac{\det\mathbf{H}(x)}{\det\mathbf{H}(y)}$. We first
find an unbiased estimator, $Y$ of the term $\log\det\mathbf{H}(x)-\log\det\mathbf{H}(y)$
which can be calculated in $\widetilde{O}\left(\nnz(\mathbf{A})+n^{2}\right)$
time using lemma \ref{lem:logdet}. We then find an
unbiased estimaor, $X$ of the determinant of $\mathbf{H}(x)$ using lemma \ref{lem:det} which describes
an algorithm to find an unbiased estimator of a value $r$ given access
to an unbiased estimator of $\log r$. 
\begin{lemma}[Determinant]
\label{lem:det}Given a random variable $Y$ with $\E(Y)=\log r$, the random variable $X$ defined as
\[
X=e\cdot\prod_{j=1}^{i}Y_{j}\mbox{ with probability }\dfrac{1}{e\cdot i!}
\]
with $Y_{j}$ being iid copies of $Y$ has $\E(X)=r.$
\end{lemma}
\begin{proof}
We know that
\[
r=\sum_{i=0}^{\infty}\dfrac{(\log(r))^{i}}{i!}.
\]
Using $X=e\cdot\prod_{j=1}^{i}Y_{j}$ with probability $\dfrac{1}{e\cdot i!}$
where $Y_{j}$ are iid random variables with $\E(Y_{j})=\log r$.
Then,
\[
\mathbb{E}[X]=\sum_{i=0}^{\infty}\frac{\E(Y)^{i}}{i!}=e^{\log(r)}=r.
\]
\end{proof}

\begin{lemma}[Log Determinant]
\label{lem:logdet} Define $\mathbf{H}(t)=\mathbf{A}^{\top}\mathbf{A}+t(\mathbf{A}^{\top}\mathbf{W}\mathbf{A}-\mathbf{A}^{\top}\mathbf{A})=\mathbf{A}^{\top}(\mathbf{I}+t(\mathbf{W}-\mathbf{I}))\mathbf{A}$.
Let $v\sim N(0,I)$ and $t$ be uniform in $[0,1]$ and
\[
Y=v^{\top}\mathbf{H}(t)^{-1}\mathbf{A}^{\top}(\mathbf{H}-I)\mathbf{A}v+\log\det\mathbf{A}^{\top}\mathbf{A}.
\]
Then, $\E(Y)=\log\det\mathbf{A}^{\top}\mathbf{W}\mathbf{A}$.
\end{lemma}
\begin{proof}
\allowdisplaybreaks
We have
\begin{align*}
&\log\det(\mathbf{H}(1))-\log\det(\mathbf{A}^{\top}\mathbf{A})\\ 
& =\int_{0}^{1}\dfrac{d\log\det\mathbf{H}(t)}{dt}dt\\
 & =\int_{0}^{1}\tr(\mathbf{H}(t)^{-1}\dfrac{d\mathbf{H}(t)}{dt})dt\\
 & =\int_{0}^{1}\tr(\mathbf{H}(t)^{-1}\mathbf{A}^{\top}(\mathbf{H}-\mathbf{I})\mathbf{A}^{\top})dt\\
 & =\mathbb{E}_{v\sim N(0,I)}[v^{\top}\int_{0}^{1}\tr(\mathbf{H}(t)^{-1}\mathbf{A}^{\top}(\mathbf{H}-\mathbf{I})\mathbf{A})dt\cdot v]\\
 & =\int_{0}^{1}\mathbb{E}_{v\sim N(0,I)}[v^{\top}\mathbf{H}(t)^{-1}\mathbf{A}^{T}(\mathbf{H}-\mathbf{I})\mathbf{A}v]dt
\end{align*}
\end{proof}
Note that given $\mathbf{H}(t)^{-1}$, we can estimate the last expression as the sum of $\mathbb{E}_{v\sim N(0,I)}[v^{\top}\mathbf{H}(t)^{-1}\mathbf{A}^{\top}(\mathbf{H}-\mathbf{I})\mathbf{A}v]$.
Maintaining $\mathbf{H}(t)^{-1}$ reduces to the inverse maintainence
problem for $\mathbf{H}$. It is shown in \citep{lee2015efficient}
that a matrix inverse can be maintained efficiently in the following
sense. Suppose we have a sequence of matrices of the form $\mathbf{A}^{\top}\mathbf{D}^{(k)}\mathbf{A}$
where each $\mathbf{D}^{(k)}$ is a slowly-changing diagonal matrix.
Then for each matrix in the sequence, its inverse times any given
vector $v$ can be computed in time $\widetilde{O}\left(\nnz(\mathbf{A})+n^{2}\right).$
We use $\mathbf{W}=\mathbf{S_{x}}^{-2}\mathbf{S_{y}^{2}}$ to calculate
as unbiased estimate of $\log\det\mathbf{H}(x)-\log\det\mathbf{H}(y)$.
\begin{lemma}[{\citep[Theorem 13]{lee2015efficient}\label{lem:inv}}]
Suppose that a sequence of matrices \\ $\mathbf{A}^{\top}\mathbf{D}^{(k)}\mathbf{A}$
for the inverse maintenance problem satisfies the $$\sum_{i}\left(\frac{d_{i}^{(k+1)}-d_{i}^{(k)}}{d_{i}^{(k)}}\right)^{2}=O(1).$$
Then there is an algorithm that with high probability maintains an
$\widetilde{O}\left(\nnz(\mathbf{A})+n^{2}\right)$-time linear system
solver for $r$ rounds in total time $\widetilde{O}\left(r(\nnz(\mathbf{A})+ n^{2}+n^{\omega})\right)$
\end{lemma}

We note that the condition $\sum_{i}\left(\frac{d_{i}^{(k+1)}-d_{i}^{(k)}}{d_{i}^{(k)}}\right)^{2}=O(1)$
is satisfied since
\begin{align*}
\sum_{i}\left(\frac{d_{i}^{(k+1)}-d_{i}^{(k)}}{d_{i}^{(k)}}\right)^{2}&=\sum_{i}\left(\frac{(s_{i}^{(k+1)})^{-2}-(s_{i}^{(k)})^{-2}}{(s_{i}^{(k)})^{-2}}\right)^{2}\\
&=O\left(\sum_{i}\left(\frac{s_{i}^{(k+1)}-s_{i}^{(k)}}{s_{i}^{(k)}}\right)^{2}\right) \\&=O\left(\|x^{(k+1)}-x^{(k)}\|_{x^{(k)}}^{2}\right).
\end{align*}
Putting these together we have the following unbiased estimator for
$\sqrt{\det\mh(x)/\det\mh(y)}$:

Compute $X=\frac{e}{2}\cdot\prod_{j=1}^{i}Y_{j}\mbox{ with probability }\dfrac{1}{e\cdot i!}$
where each $Y_{j}$ is an iid sample generated as follows:
\begin{enumerate}
\item Pick $v\sim N(0,I)$ and $t$ uniformly in $[0,1].$
\item Set $\mathbf{W}=\mathbf{S_{x}}^{-2}\mathbf{S_{y}^{2}}.$
\item Compute $Y=v^{\top}\mathbf{H}(t)^{-1}\mathbf{A}^{\top}(\mathbf{H}(1)-\mathbf{I})\mathbf{A}v$
where $\mh(t)=\mathbf{A}^{\top}(\mathbf{I}+t(\mathbf{W}-\mathbf{I}))\mathbf{A}$
using efficient inverse maintenance.
\end{enumerate}
We need one more trick. In the algorithm, at each step we need to
compute $\min\left\{ 1,\frac{p(y\rightarrow x)}{p(x\rightarrow y)}\right\} $.
While we can approximate the ratio inside the min, this might make
the overall probability incorrect due to the min function not being
smooth. So instead we propose a smoother filter. This might have other
applications.
\begin{lemma}[Smooth Metropolis filter]
Let the probability of selecting the state $y$ from the state $x$ of an ergodic Markov chain be $p(x\rightarrow y)$.
Then accepting the step $x\rightarrow y$ with probability $\dfrac{p(y\rightarrow x)}{p(y\rightarrow x)+p(x\rightarrow y)}$
gives uniform stationary distribution.
\end{lemma}
\begin{proof}
Let $\tilde{p}(x \rightarrow y)$ be the probability of taking a step from $x$ to $y$. Then, $\tilde{p}$ satisfies detailed balance.
\begin{align*}
\tilde{p}(x\rightarrow y) & =p(x\rightarrow y)\cdot\dfrac{p(y\rightarrow x)}{p(y\rightarrow x)+p(x\rightarrow y)}\\
 & =\dfrac{p(x\rightarrow y)p(y\rightarrow x)}{p(y\rightarrow x)+p(x\rightarrow y)}\\
 & =p(y\rightarrow x)\cdot\dfrac{p(x\rightarrow y)}{p(y\rightarrow x)+p(x\rightarrow y)}\\
 & =\tilde{p}(y\rightarrow x)
\end{align*}
So, $\tilde{p}(x\rightarrow y)=\tilde{p}(y\rightarrow x)$ for all
$x$ and $y$. Hence the stationary distribution is uniform.
\end{proof}

For the Dikin walk, $\frac{p(y\rightarrow x)}{p(x\rightarrow y)}=\sqrt{\frac{\det(\mh_{x})}{\det(\mh_{y})}}$.
Note that the rejection probability function $\frac{p(y\rightarrow x)}{p(y\rightarrow x)+p(x\rightarrow y)}\allowbreak=\frac{\frac{p(y\rightarrow x)}{p(x\rightarrow y)}}{1+\frac{p(y\rightarrow x)}{p(x\rightarrow y)}}$
is increasing in $\frac{p(y\rightarrow x)}{p(x\rightarrow y)}$.
As Dikin barrier is strongly self-concordant (Lemma \ref{lem:dssc})
and by \eqref{eq:rejection}, we get that with probability at least
$0.99$, for $y$ randomly drawn from $E_{x}$, $\frac{\vol(E_{x}(r))}{\vol(E_{y}(r))}\geq0.9922$.
Hence, the probability of not rejecting at each step at least $0.498$ with large probability.
\begin{proof}[Proof of Theorem \ref{thm:dikin}]
Implementing Dikin walk requires mai\-ntaining matrices $\mathbf{H}_{t}=\mathbf{A}^{\top}\mathbf{S}_{t}^{-2}\mathbf{A}$
corresponding to point $x_{t}$. \ref{lem:inv} shows that this can
be done in $\widetilde{O}\left(n^{\omega}+r(nnz(\mathbf{A})+n^{2})\right)$
time where $r$ is the number of steps in the chain. Additionaly,
each step requires calculating the rejection probability which is
a smooth function in $\dfrac{\det(\mathbf{H}_{t})}{\det(\mathbf{H}_{t+1})}$
and hence can be calculated in $\widetilde{O}\left(nnz(\mathbf{A})+n^{2}\right)$
amortized time using lemmas \ref{lem:det} and \ref{lem:logdet}.
\end{proof}

\section{Strong Self-Concordance of other barriers}\label{section5}

Here we analyze the strong self-concordance of the universal and entropic
barriers.
\begin{proof}[Proof of Lemma \ref{lem:SSC-KLS}.]
The entropic barrier is the dual of 
\[
f(\theta)=\log(\int_{x\in K}\exp(\theta^{\top}x)dx).
\]
Then, its the first three derivatives are moments \citep{bubeck2014entropic}:
\begin{align*}
Df(\theta)[h_{1}] & =\frac{\int_{x\in K}x^{\top}h_{1}\exp(\theta^{\top}x)dx}{\int_{x\in K}\exp(\theta^{\top}x)dx}\\
 & =\E_{x\sim p_{\theta}}x^{\top}h_{1}.
\end{align*}
where $p_{\theta}$ is the corresponding exponential distribution
with support $K$.
\begin{align*}
D^{2}f(\theta)[h_{1},h_{2}] & =\frac{\int_{x\in K}x^{\top}h_{1}x^{\top}h_{2}\exp(\theta^{\top}x)dx}{\int_{x\in K}\exp(\theta^{\top}x)dx}\\
&\quad -\frac{\prod_{i=1}^2\int_{x\in K}x^{\top}h_{i}\exp(\theta^{\top}x)dx}
{\left(\int_{x\in K}\exp(\theta^{\top}x)dx\right)^{2}}\\
 & =\E_{x\sim p_{\theta}}h_{2}^{\top}xx^{\top}h_{1}-h_{2}^{\top}\mu\mu^{\top}h_{1}\\
 & =\E_{x\sim p_{\theta}}(x-\mu)^{\top}h_{1}\cdot(x-\mu)^{\top}h_{2}
\end{align*}
Next, we note that
\begin{align*}
D\mu[h] & =D\frac{\int_{x\in K}x\exp(\theta^{\top}x)dx}{\int_{x\in K}\exp(\theta^{\top}x)dx}[h]\\
 & =\frac{\int_{x\in K}xx^{\top}h\exp(\theta^{\top}x)dx}{\int_{x\in K}\exp(\theta^{\top}x)dx}\\
 &\quad -\frac{\int_{x\in K}x\exp(\theta^{\top}x)dx}{\int_{x\in K}\exp(\theta^{\top}x)dx}\cdot\frac{\int_{x\in K}x^{\top}h\exp(\theta^{\top}x)dx}{\int_{x\in K}\exp(\theta^{\top}x)dx}\\
 & =\E_{x\sim p_{\theta}}xx^{\top}h-\mu\mu^{\top}h\\
 & =\E_{y\sim p_{\theta}}(y-\mu)(y-\mu)^{\top}h.
\end{align*}
So, we have
\begin{align*}
&D^{3}f(\theta)[h_{1},h_{2},h_{3}]\\
=&  \E_{x\sim p_{\theta}}(-\E_{y\sim p_{\theta}}(y-\mu)(y-\mu)^{\top}h_{3})^{\top}h_{1}\cdot(x-\mu)^{\top}h_{2}\\
 & +\E_{x\sim p_{\theta}}(x-\mu)^{\top}h_{1}\cdot(-\E(y-\mu)(y-\mu)^{\top}h_{3})^{\top}h_{2}\\
 & +\E_{x\sim p_{\theta}}(x-\mu)^{\top}h_{1}\cdot(x-\mu)^{\top}h_{2}\cdot(x-\mu)^{\top}h_{3}\\
= & \E_{x\sim p_{\theta}}(x-\mu)^{\top}h_{1}\cdot(x-\mu)^{\top}h_{2}\cdot(x-\mu)^{\top}h_{3}.
\end{align*}
By \citep[(2.15)]{nesterov1994interior}, we have that
\[
D^{2}f^{*}(x_{\theta})[h_{1},h_{2}]=h_{1}^{\top}\nabla^{2}f(\theta)^{-1}h_{2}
\]
and
\begin{align*}
   &D^{3}f^{*}(x_{\theta})[h_{1},h_{2},h_{3}]\\&\quad=-D^{3}f(\theta)[\nabla^{2}f(\theta)^{-1}h_{1},\nabla^{2}f(\theta)^{-1}h_{2},\nabla^{2}f(\theta)^{-1}h_{3}]
\end{align*}

where $x_{\theta}=\nabla f(\theta)$. Hence, we have
\begin{align*}
 & \nabla^{2}f^{*}(x_{\theta})^{-\frac{1}{2}}D^{3}f^{*}(x_{\theta})[h]\nabla^{2}f^{*}(x_{\theta})^{-\frac{1}{2}}\\
&=  -\E_{x\sim p_{\theta}}\nabla^{2}f(\theta)^{-\frac{1}{2}}(x-\mu)(x-\mu)^{\top}\nabla^{2}f(\theta)^{-\frac{1}{2}}\\
&\qquad \cdot(x-\mu)^{\top}\nabla^{2}f(\theta)^{-1}h\\
&= -\E_{x\sim\tilde{p}_{\theta}}xx^{\top}\cdot x^{\top}\nabla^{2}f(\theta)^{-\frac{1}{2}}h
\end{align*}
where $\tilde{p}_{\theta}$ is the distribution given by $\nabla^{2}f(\theta)^{-\frac{1}{2}}(x-\mu)$
where $x\sim p_{\theta}$. Note that $\tilde{p}_{\theta}$ is isotropic
and \citep[Fact 6.1]{eldan2013thin} shows that 
\begin{equation}
\max_{\|v\|_{2}=1}\norm{\E_{x\sim\tilde{p}_{\theta}}xx^{T}(x^{T}v)}_{F}=O(\psi_{n}).\label{eq:K_n_psi_n}
\end{equation}

Hence, we have that
\begin{align*}
   &\norm{\nabla^{2}f^{*}(x_{\theta})^{-\frac{1}{2}}D^{3}f^{*}(x_{\theta})[h]\nabla^{2}f^{*}(x_{\theta})^{-\frac{1}{2}}}_{F}\\
   &\quad =O(\psi_{n})\norm{\nabla^{2}f^{*}(x_{\theta})^{-\frac{1}{2}}h}_{2}=O(\psi_{n})\|h\|_{x_{\theta}}.
\end{align*}

This proves the lemma for the entropic barrier (recall that the entropic
barrier is $f^{*}$ instead of $f$).

For the universal barrier, first we recall that the polar of a convex
set $K$ is $K^{\circ}(x)=\left\{ z:z^{\top}(y-x)\le1\quad\forall y\in K\right\} $
and the barrier function is 
\[
\Phi(x)=\log\vol(K^{\circ}(x)).
\]
Its derivatives have the following identities \citep[Page 52]{nesterov1994interior}.
Here the random point $y$ is drawn uniformly from the polar $K^{\circ}(x)$.
\begin{align*}
\nabla^{2}\Phi(x)= & (n+2)(n+1)\E yy{}^{\top}-(n+1)^{2}\E y\E y^{\top},\\
D\nabla^{2}\Phi(x)[h]= & -(n+1)(n+2)(n+3)\E yy^{\top}(y^{\top}h)\\
&+(n+1)^{2}(n+2)\E yy^{\top}\cdot\E y^{\top}h\\
 & +2(n+1)^{2}(n+2)\E y(y^{\top}h)\cdot\E y^{\top}\\
 &-2(n+1)^{3}\E y\cdot\E y^{\top}\cdot\E y^{\top}h
\end{align*}
Let $\mu=\E y$, we can re-write the derivatives as follows:
\begin{align*}
\nabla^{2}\Phi(x)= & (n+2)(n+1)\E(y-\mu)(y-\mu)^{\top}+(n+1)\mu\mu^{\top}\\
D\nabla^{2}\Phi(x)[h]= & -\Pi_{i=1}^3 (n+i)\E(y-\mu)(y-\mu)^{\top}(y-\mu)^{\top}h\\
 & -2(n+2)(n+1)(\E(y-\mu)(y-\mu)^{\top}\mu^{\top}h\\
 & +\E\mu(y-\mu)^{\top}(y-\mu)^{\top}h+\E(y-\mu)\mu^{\top}(y-\mu)^{\top}h)\\
 & -2(n+1)\mu\mu^{\top}\mu^{\top}h.
\end{align*}
Without loss of generality, we assume $\nabla^{2}\Phi(x)=I$. Then,
we have
\[
(n+2)(n+1)\E(y-\mu)(y-\mu)^{\top}\preceq I\qquad\text{and}\qquad(n+1)\mu\mu^{\top}\preceq I.
\]
For the first term, \eqref{eq:K_n_psi_n} shows that
\[
\|(n+1)(n+2)(n+3)\E(y-\mu)(y-\mu)^{\top}(y-\mu)^{\top}h\|_{F}=O(\psi_{n}).
\]
The Frobenius norm of next three terms are bounded by
\[
2\left|\mu^{\top}h\right|\norm{(n+2)(n+1)\E(y-\mu)(y-\mu)^{\top}}_{F}\le2\sqrt{n}\norm{\mu}\le2
\]
and so is the last term: 
\[
2\norm{(n+1)\mu\mu^{\top}}_{F}\left|\mu^{\top}h\right|\le2.
\]
\end{proof}
To conclude this section, we remark that the universal and entropic barriers do \emph{not }satisfy our symmetry condition. Consider a rotational cone $C=\left\{ x\,:\,\sum_{i=2}^{n}x_{i}^{2}\le x_{1}^{2},0\le x_{1} \le 1\right\} $ and any point $x=(x_{1},0,\ldots,$0). Then symmetric body
around $x$, namely $K=C\cap(x-C)$ has the property that (a) the
John ellipsoid satisfies $E\subset K\subset\sqrt{n}C$ (as it does
for any symmetric convex body) and (b) the inertial ellipsoid has
a sandwiching ratio of $n$, proving that $\bar{\nu}\ge n=\Omega(\nu^{2}).$
For the entropic barrier, we have a similar result because multiplying
the indicator function of this symmetric convex body with an exponential
function of the form $e^{-c^{T}x}$ still has the same property for
the inertial ellipsoid. This example highlights the advantages of
barriers with John-like ellipsoids (log barrier, LS barrier) vs Inertia-like
ellipsoids (universal, entropic).

\bibliographystyle{alpha}
\bibliography{arxiv}
\appendix
\section{Proofs}
\subsection{Proof of Lemma \ref{lem:global_self_concordant}}\label{proof:2}
\begin{proof}
Let $h=y-x$, $x_{t}=x+th$ and $\phi(t)=h^{\top}\mh(x_{t})h$.
Then, 
\[
\left|\phi'(t)\right|=\left|h^{\top}\frac{d}{dt}\mh(x_{t})h\right|\leq2\|h\|_{x_{t}}^{3}=2\phi(t)^{3/2}.
\]
Hence, we have $\left|\frac{d}{dt}\frac{1}{\sqrt{\phi(t)}}\right|\leq1$.
Therefore, $\frac{1}{\sqrt{\phi(t)}}\geq\frac{1}{\sqrt{\phi(0)}}-t$
and, 
\begin{equation}
\phi(t)\leq\frac{\phi(0)}{(1-t\sqrt{\phi(0)})^{2}}.\label{eq:phi_bound}
\end{equation}

Now we fix any $v$ and define $\psi(t)=v^{\top}\mh(x_{t})v$. Then,
\[
\left|\psi'(t)\right|=\left|v^{\top}\frac{d}{dt}\mh(x_{t})v\right|\leq2\|h\|_{x_{t}}\|v\|_{x_{t}}^{2}=2\phi(t)\psi(t).
\]
Using \eqref{eq:phi_bound} at the end, we have
\[
\left|\frac{d}{dt}\ln\psi(t)\right|\leq\frac{2\sqrt{\phi(0)}}{(1-t\sqrt{\phi(0)})}.
\]
Integrating both sides from $0$ to $1$, 
\[
\left|\ln\frac{\psi(1)}{\psi(0)}\right|\leq\int_{0}^{1}\frac{2\sqrt{\phi(0)}}{(1-t\sqrt{\phi(0)})}dt=2\ln(\frac{1}{1-\sqrt{\phi(0)}}).
\]
The result follows from this with $\psi(1)=v^{\top}\mh(y)v$, $\psi(0)=v^{\top}\mh(x)v$,
and $\phi(0)=\|x-y\|_{x}^{2}$.
\end{proof}

\subsection{Proof of Lemma \ref{lem:global_strongly_self_concordant}}\label{proof:4}
\begin{proof}
Let $x_{t}=(1-t)x+ty$. Then, we have
\begin{align*}
&\|\mh(x)^{-\frac{1}{2}}(\mh(y)-\mh(x))\mh(x)^{-\frac{1}{2}}\|_{F} \\
&\quad =\int_{0}^{1}\|\mh(x)^{-\frac{1}{2}}\frac{d}{dt}\mh(x_{t})\mh(x)^{-\frac{1}{2}}\|_{F}dt.
\end{align*}

We note that $\mh$ is self-concordant. Hence, Lemma \ref{lem:global_self_concordant}
shows that
\begin{align*}
&\|\mh(x)^{-\frac{1}{2}}\frac{d}{dt}\mh(x_{t})\mh(x)^{-\frac{1}{2}}\|_{F}^{2}\\
&\quad =\tr\mh(x)^{-1}\left(\frac{d}{dt}\mh(x_{t})\right)\mh(x)^{-1}\left(\frac{d}{dt}\mh(x_{t})\right)\\
 &\quad \leq\frac{1}{\left(1-\Vert x-x_{t}\Vert_{x}\right)^{4}}\tr\mh(x_{t})^{-1}\left(\frac{d}{dt}\mh(x_{t})\right)\mh(x_{t})^{-1}\left(\frac{d}{dt}\mh(x_{t})\right)\\
 &\quad \leq\frac{4}{\left(1-\Vert x-x_{t}\Vert_{x}\right)^{4}}\|x-x_{t}\|_{x_{t}}^{2}\\
 &\quad \leq\frac{4}{\left(1-\Vert x-x_{t}\Vert_{x}\right)^{6}}\|x-x_{t}\|_{x}^{2}
\end{align*}

where we used the strong self-concordance in the second inequality and Lemma \ref{lem:global_self_concordant}
again for the last inequality. Hence, 
\begin{align*}
\|\mh(x)^{-\frac{1}{2}}(\mh(y)-\mh(x))\mh(x)^{-\frac{1}{2}}\|_{F} & \leq\int_{0}^{1}\frac{2\|x-x_{t}\|_{x}}{\left(1-\Vert x-x_{t}\Vert_{x}\right)^{3}}dt\\
 & =\int_{0}^{1}\frac{2t\|x-y\|_{x}}{(1-t\|x-y\|_{x})^{3}}dt\\
 & =\frac{\|x-y\|_{x}}{(1-\|x-y\|_{x})^{2}}.
\end{align*}
\end{proof}

\end{document}